\documentclass[11pt]{article}

\usepackage[utf8]{inputenc} 
\usepackage[T1]{fontenc}    
\usepackage{hyperref}       
\usepackage{url}            
\usepackage{booktabs}       
\usepackage{amsfonts}       
\usepackage{nicefrac}       
\usepackage{microtype}      

\usepackage{multirow}
\usepackage{amstext}
\usepackage{amsmath}
\usepackage{amsthm}
\usepackage{amssymb}
\usepackage{bm}
\usepackage{wrapfig}
\usepackage{amsfonts, amsmath}
\usepackage[utf8]{inputenc}
\usepackage{graphicx}
\usepackage{bbm, comment}
\usepackage{subfigure}
\usepackage{color}
\usepackage{float}
\usepackage{wrapfig}
\usepackage{enumitem}
\usepackage{url}
\usepackage{MnSymbol,wasysym,marvosym,ifsym}

\usepackage[ruled]{algorithm2e} 

\makeatletter
\def\BState{\State\hskip-\ALG@thistlm}
\makeatother

\usepackage{scalerel,stackengine}
\stackMath
\newcommand\reallywidehat[1]{%
\savestack{\tmpbox}{\stretchto{%
  \scaleto{%
    \scalerel*[\widthof{\ensuremath{#1}}]{\kern-.6pt\bigwedge\kern-.6pt}%
    {\rule[-\textheight/2]{1ex}{\textheight}}
  }{\textheight}%
}{0.5ex}}%
\stackon[1pt]{#1}{\tmpbox}%
}
\usepackage{tabulary}
\usepackage{booktabs}

\usepackage[top=1.in, bottom=1.in, left=1.in, right=1.in]{geometry}
\usepackage[numbers]{natbib}

\newtheorem*{theorem*}{Theorem}
\newtheorem{theorem}{Theorem}[section]
\newtheorem{lemma}[theorem]{Lemma}

\newtheorem{corollary}[theorem]{Corollary}
\newtheorem{example}[theorem]{Example}

\newtheorem{definition}[theorem]{Definition}

\newcommand{\dmi}{\textsc{DMI}}

\newcommand{\dmic}{\text{DMI-C}}
\newcommand{\dmis}{\text{DMI-S}}

\newcommand{\argmax}{\mathrm{argmax}}

\newcommand{\idxmax}{\mathrm{idxmax}}

\newcommand{\ncol}{\mathrm{ncol}}

\newcommand{\diag}{\mathrm{diag}}

\SetCommentSty{mycommfont}

\SetKwInput{KwInput}{Input}                
\SetKwInput{KwOutput}{Output}              

\begin{document}
\title{Information Elicitation Meets Clustering}
\author{Yuqing Kong \\The Center on Frontiers of Computing Studies,\\Peking University \\\texttt{yuqing.kong@pku.edu.cn} \\}
\date{}
\maketitle

\begin{abstract}
	In the setting where we want to aggregate people's subjective evaluations, plurality vote may be meaningless when a large amount of low-effort people always report ``good'' regardless of the true quality. ``Surprisingly popular'' method, picking the most surprising answer compared to the prior, handle this issue to some extent. However, it is still not fully robust to people's strategies. Here in the setting where a large number of people are asked to answer a small number of multi-choice questions (multi-task, large group), we propose an information aggregation method which is robust to people's strategies. Interestingly, this method can be seen as a rotated ``surprisingly popular''. It is based on a new clustering method, Determinant MaxImization (DMI)-clustering, and a key conceptual idea that information elicitation without ground-truth can be seen as a clustering problem. Of independent interest, DMI-clustering is a general clustering method that aims to maximize the volume of the simplex consisting of each cluster's mean multiplying the product of the cluster sizes. We show that DMI-clustering is invariant to any non-degenerate affine transformation for all data points. When the data point's dimension is a constant, DMI-clustering can be solved in polynomial time. In general, we present a simple heuristic for DMI-clustering which is very similar to Lloyd's algorithm for k-means. Additionally, we also apply the clustering idea in the single-task setting and use spectral method to propose a new aggregation method that utilizes the second-moment information elicited from the crowds.

\end{abstract}


\section{Introduction}

When we want to forecast an event, make a decision, evaluate a product, or test a policy, we need to elicit information from the crowds and then aggregate the information. In many scenarios, we may not have ground-truth answer to verify the elicited information. For example, when we poll people's voting intentions or elicit people's subjective evaluations, there is no ground-truth for people's feedback. Recently, a line of work \cite{MRZ05,prelec2004bayesian,2016arXiv160303151S,Kong:2019:ITF:3309879.3296670,liuchen,kong2020dominantly} explore how to incentivize people to be honest without ground-truth. Some work \cite{prelec2017solution,liuchen} also study how to aggregate the information without ground-truth. 

A natural way to aggregate information without ground-truth is ``plurality vote'', i.e., output the most popular option. For example, in product evaluation, when the feedback of is 56\% ``good'', 40\% ``so so'', 4\% ``bad'', the plurality answer is ``good''. However, if there are a large amount of people who always choose ``good'' regardless of the quality, plurality vote will always output ``good'' thus becomes meaningless. A more robust method is ``surprisingly popular'', proposed by \citet{prelec2017solution}. In the above example, if the prior is 70\% ``good'', 26\% ``so so'', 4\% ``bad'', ``surprisingly popular'' will aggregate the feedback to ``so so'' which is the most surprisingly popular answer ($\frac{40}{26}>\frac{56}{70},\frac44$). This will handle the above issue to a certain degree. However, ``surprisingly popular'' may not work in some scenarios. 

\begin{example} 
	When the ground truth state is ``good'', the distribution over people's feedback is (56\% ``good'', 40\% ``so so'', 4\% ``bad''); when the ground truth state is ``so so'', the distribution is (56\% ``good'', 34\% B, 10\% ``bad''); ground truth state ``bad'' corresponds to (46\% ``good'', 44\% ``so so'', 10\% ``bad''). Each state shows up with equal probability. Thus, the prior is approximately (52.7\% ``good'',39.3\% ``so so'',8\% ``bad'')
\end{example} 

In the above example, for both (56\% ``good'', 34\% ``so so'', 10\% ``bad'') and (46\% ``good'', 44\% ``so so'', 10\% ``bad''), the ``surprisingly popular'' answer is ``bad'', while they correspond to different ground truth. 

Therefore, ``surprisingly popular'' requires an additional prior assumption about the relationship between the ground truth state and people's answers \cite{prelec2017solution}. However, even if people's honest answers satisfy the assumption, people can perform a strategy to their answers to break the assumption. Here is an example. 

\begin{example}
	When the ground truth state is ``good'', the distribution over people's feedback is (100\%,0\%,0\%); when the ground truth state is ``so so'', the distribution is (0\%,100\%,0\%); ground truth state ``bad'' corresponds to (0\%,0\%,100\%). All people perform a strategy: given my private answer is ``good'', I will report ``good'' with probability .56, ``so so'' with probability .4, ``bad'' with probability .04, given my answer is ``so so'', I will ... The strategy is represented by a row-stochastic matrix $\begin{bmatrix}
		.56&.4&.04\\
		.56& .34&.1\\
		.46 &.44&.1
	\end{bmatrix}$ where each row represents the reporting distribution for each private answer. In this case, people's strategic behaviors change the report distribution to the previous example and "surprisingly popular" will not work here. 
\end{example}

Thus, a fundamental problem of the ``surprisingly popular'' method is that though it is more robust than plurality vote method, it is not fully robust to people's strategies. 

To this end, we can naturally ask for an information aggregation method that is robust to people's strategies, that is, regardless of people's strategies, the aggregation result will not change. 

However, it may ask too much since people can permute their answers and no method is robust to this permutation strategy. Thus we change the requirement to the aggregation result will not change up to a permutation and use side information (e.g. ground truth of a small set of questions) to determine the permutation.

\begin{figure}
\center
	\includegraphics[width=9cm]{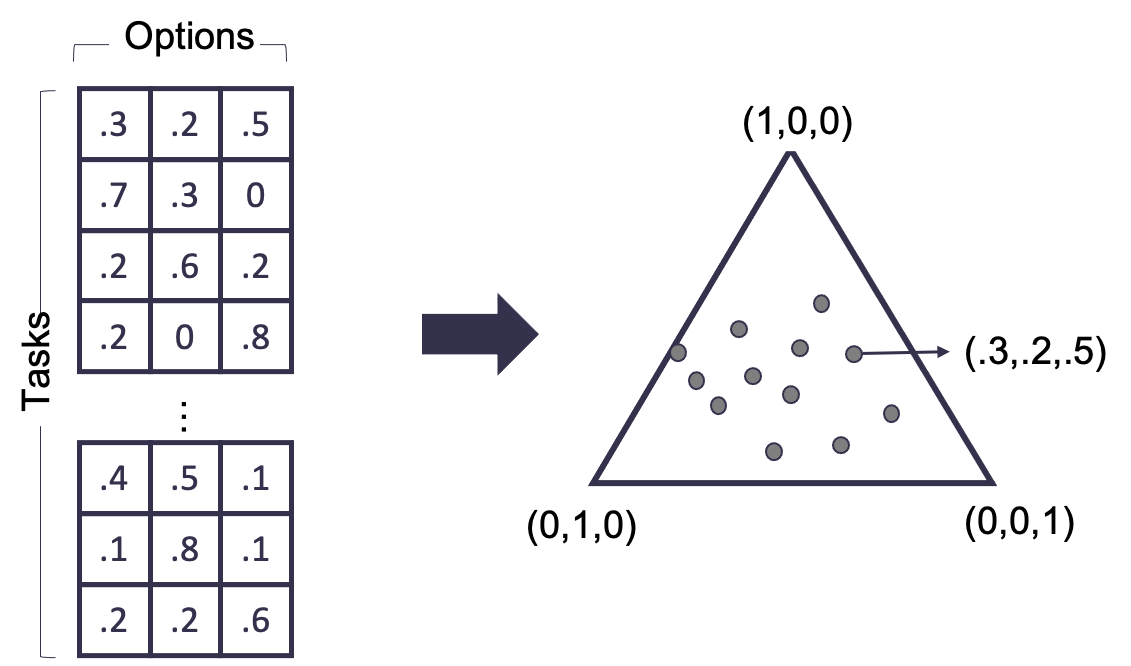}
	\caption{Reduction to clustering: we map each task to a data point that represents the task's collected feedback (e.g. 30\% ``good'', 20\% ``so so'', 50\% ``bad'') and then clustering the data points to assign each task a final answer.}\label{fig:reduction}
\end{figure}

It turns out that when we have a large number of people, and a small set of a priori similar tasks (each task is a multi-choice question), the answer is yes. A key conceptual idea is to reduce information aggregation without ground truth to a clustering/unsupervised learning problem. That is, we can use a datapoint to represents a task's feedback. Assigning answers to the tasks can be understood as clustering the data points (Figure~\ref{fig:reduction}). Recall the previous example, when people perform a strategy to their reports, all feedback are transformed by an affine transformation \footnote{It holds even if people adopt different strategies as long as each task is performed by a large number random people.}. Therefore, designing a robust information aggregation method is equivalent to design an affine-invariant clustering method. We then propose an affine-invariant clustering method, Determinant MaxImization (DMI) clustering, and apply it into the information elicitation without verification. 

DMI-clustering is the technical highlight in this paper. Previous affine-invariant clustering work \cite{begelfor2006affine,fitzgibbon2002affine} either has a different definition for affine-invariance\footnote{They require that items that are same up to affine transformations (e.g. images of the same chair observed in different angles) should be clustered into the same cluster. In our setting, we require the clustering result remains the same when all items are transformed by the same affine transformation.} from this work or has an underlying probabilistic model assumption and requires sufficiently number of data points \cite{brubaker2008isotropic,huang2020affine}. In contrast, DMI-clustering is a non-parametric affine-invariant method that applies to any finite number of data points and does not require any probabilistic model assumption. DMI-clustering provides a relatively affine-invariant evaluation to each clustering, the volume of the simplex consisting of each cluster mean multiplying the product of the cluster sizes, and aims to find the clustering with the highest evaluation. 

When we apply DMI-clustering into information elicitation, the dimension of the data points is the number of options, which can be seen as a constant. In this case, we show that there exists a polynomial time algorithm for DMI-clustering. In general, we present a simple local-maximal searching heuristic for DMI-clustering, k-cofactors, which is similar to Lloyd's algorithm \cite{lloyd1982least} for k-means. We illustrate the geometric interpretation of DMI-clustering and k-cofactors for 2d points in Figure~\ref{fig:dmigoal}. 

To apply DMI-clustering to multi-task information elicitation, we can normalize each question's feedback as a data point (e.g. (46\% A, 44\% B, 10\% C)). Given multiple questions' feedback, we can apply DMI-clustering to obtain the final answer\footnote{To determine the cluster's label, we can use side information like ground-truth of a small set of questions.}. 

We also compare DMI-clustering with the surprisingly popular method. Interestingly, we find that they are equivalent for binary-question and in general, DMI-clustering can be understood as a ``rotated'' surprisingly popular method. 

\begin{figure}
	\includegraphics[width=16cm]{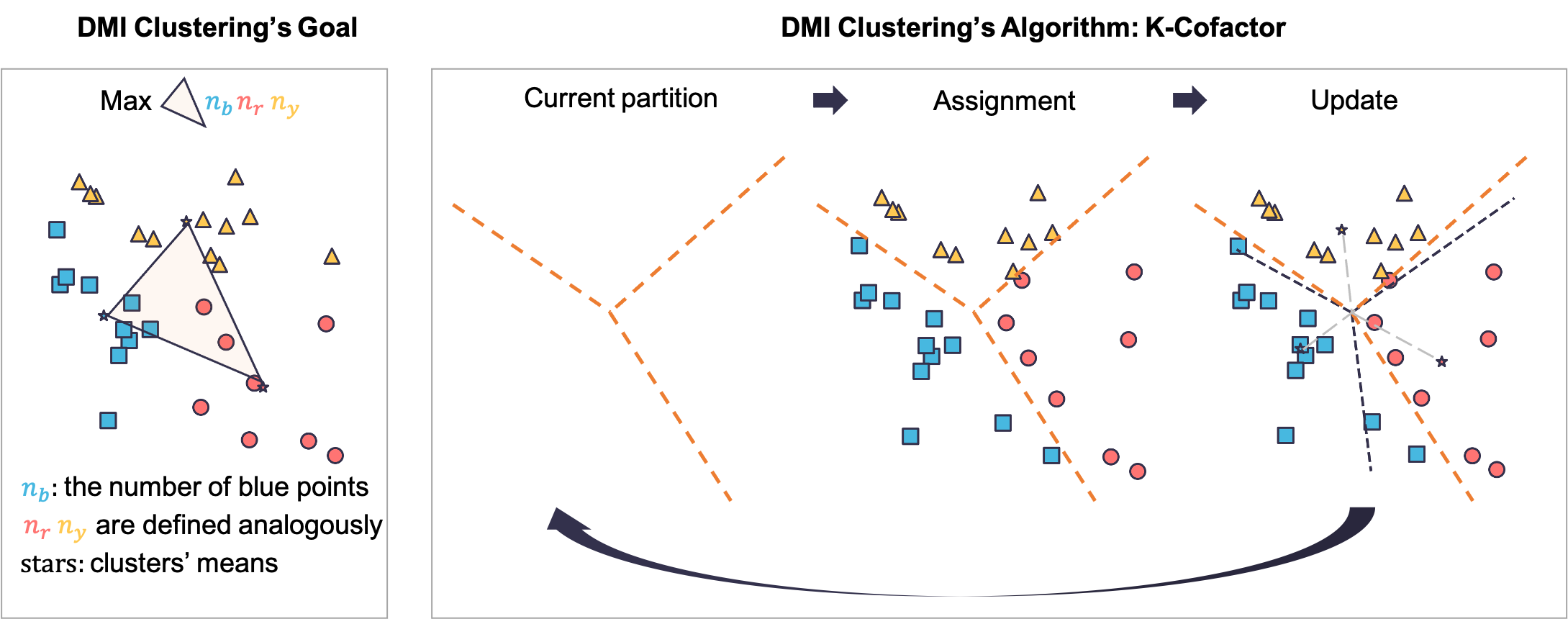}
	\caption{Geometric interpretation of DMI-clustering's goal and algorithm in 2d. Left: DMI-clustering aims to maximize the area of the triangle that consists of the clusters' means, multiplying the product of the cluster sizes. Right: we can update the current partition (orange dashed lines) by 1) assigning the points into clusters based on the current partition; 2) extending the line segmentations of each cluster's mean (star) and all points' mean (grey dashed lines) to obtain the new partition (black dashed lines). When we apply DMI-clustering to multi-task information elicitation, a point here represents a task's collected feedback.}\label{fig:dmigoal}
\end{figure}

In addition to the multi-task setting, it's also very useful to think from a clustering view in the single-task setting where we only have a single multi-choice question to people. The ``surprisingly popular'' method can be seen as a clustering method which only employs the first moment. To obtain the first moment, we can ask people's second moment information (e.g. what's your prediction for other people's answer) and reconstruct the first moment using the second moment~\cite{prelec2017solution}. From a clustering view, it's very natural to apply spectral method, which uses the second moment information, to design a new information aggregation method in the single-task case. If we can use a hypothetical question to elicit a higher-order-moment information, we can apply higher-order-moment clustering method for single-task setting.

\paragraph{Our results}  We connect information elicitation without ground truth with clustering, and 
\begin{description}
	\item [Multi-task \& Affine-invariance] we propose an affine-invariant clustering method and apply it to propose a multi-task information elicitation mechanism that 1) aggregates information in a way that is robust to people's strategies when the number of agents is infinite, 2) reduces the payment variance in some scenarios while still preserves previous multi-task mechanisms' incentive properties;
	\item [Single-task \& Methods of moments] we apply method of moments clustering method in the single-task setting. With only the second-moment information, we apply spectral method to propose a new aggregation mechanism which is better than the ``surprisingly popular'' method in some scenarios. 
\end{description}

The main conceptual contribution of this work is the connection between clustering and information elicitation without ground truth. This connection not only brings clustering techniques to information elicitation without verification, but also brings attention to the clustering setting where the input data point, and the property of the data (e.g. moment information) can be elicited from the crowds. We hope this connection can generate more results in the future for information elicitation/aggregation without ground truth in the future.




\subsection{Related Work}

Information elicitation without ground truth (peer prediction) is proposed by \citet{MRZ05}. The core problem is how to incentivize people to report high-quality answer even if the answer cannot be verified. The key idea is to reward each agent based on her report and her peer's reports. A series of work \cite{prelec2004bayesian,dasgupta2013crowdsourced, 2016arXiv160303151S,Kong:2019:ITF:3309879.3296670,liuchen,kong2020dominantly} have focused on this problem and develop reward schemes in varies of settings. However, in many scenarios, we want to not only  incentivize high-quality answer but also make an aggregation of the answers. Thus, Some of information elicitation work also explore how to aggregate information without ground truth and there are two approaches here. 

One approach is to use the plurality of high-payment agents' answer. \citet{liuchen} employ this way to aggregate people's reports based on various of payment schemes. However, incentive-compatible payment schemes only guarantees high-quality answer receives high payment in expectation. If the payment has a high variance, which happens especially when the number of tasks is small, high-payment agent may not provide high quality answers. Thus, here each agent needs to perform a sufficiently number of tasks to reduce the payment variance. A large literature of extracting wisdom from crowd also uses the idea that first estimate the crowd's quality and then use the estimated quality to aggregate their opinions (e.g. (e.g. \cite{budescu2015identifying, weller2007cultural})). The quality can either be estimated from previous performance the crowds or from a EM framework proposed by \citet{dawid1979maximum}. Though the EM framework based estimation does not requires historical data, it still requires each person to answer a sufficient amount questions to get a good estimation. Moreover, it usually requires an underlying probabilistic model: there exists a ground truth and conditioning on the ground truth, people's reports are independent.


Another approach is to directly evaluate the answer itself. A popular method is the ``surprisingly popular'' method. \citet{prelec2017solution} uses this method in the single-task setting by eliciting people's predictions for other people's answer and reconstructing the prior using the predictions. In the setting we have already known the prior, this method can be directly implemented. However, like we mentioned in introduction this method requires an additional assumption for the prior. 

More importantly, none of these aggregation approaches is robust to people's strategies. Though for Determinant based Mutual Information (DMI)-Mechanism \cite{kong2020dominantly}, the \emph{expected} payment is relatively robust to other people's strategies, we still need each agent to perform a sufficiently number of tasks to reduce the payment variance. However, the setting when we have a large number of agents and each agent performs a small number of tasks may be more practical. The multi-task setting in this work focuses on this more practical setting and proposes an aggregation method, DMI-clustering, which is robust to people's strategies. DMI-clustering is inspired by DMI-Mechanism \cite{kong2020dominantly}. Unlike DMI-Mechanism based aggregation approach uses the plurality of high-payment agents' answer, DMI-clustering learns the answers which has the highest determinant mutual information with the collection of people's answers. When we have a large number of participants, the learning process is independent of their strategies, even if each participant only performs a small number of tasks. \citet{xu2019l_dmi} also employ the property of Determinant based Mutual Information to design a robust loss function. Unlike this work, \citet{xu2019l_dmi} focus on the noisy labels setting rather than the unsupervised clustering setting.

DMI-clustering is a new general clustering method. The most important property of DMI-clustering is affine-invariance. Affine-invariant clustering is well studied in image classification \cite{werman1995similarity,fitzgibbon2002affine,begelfor2006affine}. However, they have a different definition for affine-invariance from this work. In their definition, images that are about the same item observed in different angles should be classified into the same cluster. 

Unlike the above work about affine-invariance, we focus on the setting where all data points are transformed by the same affine transformation and we want the clustering for these data points remains invariant to the transformation. \citet{brubaker2008isotropic, huang2020affine} also use this definition and propose an affine-invariant clustering method. However, \citet{brubaker2008isotropic} focus on the application in mixture model and require the knowledge of the number of components and sufficient amount of sample data points. \citet{huang2020affine} have an underlying probabilistic model and develop a maximum likelihood estimator based algorithm. In contrast, DMI-clustering is a non-parametric affine-invariant method that applies to any finite number of data points and does not require any probabilistic model assumption. Moreover, DMI-clustering provides a relatively affine-invariant evaluation to each clustering. Thus, when we have side information that gives a constraint for the set of clustering (e.g. we need to pick one from clustering I or clustering II), DMI-clustering method is still affine-invariant. 

DMI-clustering is more similar to k-means clustering. Like k-means, it has a natural geometric interpretation and a simple local-maxima searching heuristic that uses each cluster's mean. One difference from k-means is that DMI-clustering can not pick any $k$ it wants since it needs to satisfy the affine-invariance\footnote{If we have $d$ linearly independent $d$-dimensional data points, k-means method can pick any $k$. However, DMI-clustering must pick $k=d$ and make each individual data point a cluster by itself since these data points can be any other linearly independent $d$-dimensional data points by a proper affine-transformation. }. 

The definition of DMI-clustering uses determinant. Log-determinant divergence \cite{cichocki2015log}'s definition also uses determinant and is affine-invariant in its definition. However, still, the affine-invariance of log-determinant divergence is different from this work. It refers to the property that the divergence between two positive definite matrices is invariant to affine transformations. The geometric interpretation of DMI-clustering is maximizing the volume of simplex constituted by cluster's means, multiplying the product of cluster sizes. \citet{thurau2010yes} also connect simplex volume maximization with clustering. However, their problem is very different from ours and irrelevant to affine-invariance. They aim to design an efficient algorithm for matrix factorization by greedily picking the basis vectors such that the simplex formed by these basis vectors has the maximal volume. 

\section{Determinant MaxImization (DMI) Clustering}

In this section, we will introduce a new clustering method, DMI-clustering, that is based on determinant maximization. We will show that this method is invariant to any affine transformation which is performed on all data points due to the multiplicative property of determinant. We will also show that this clustering method will generate a clustering which has a natural geometric interpretation: there will be a representative vector for each cluster and each data point is mapped to the cluster whose representative vector has the highest inner-product with the data point. 

Moreover, we will introduce a heuristic algorithm, k-cofactors, which is similar to Lloyd's algorithm for k-means. When the data points are in 2-dimension, we visualize the algorithm and introduce an interesting geometric way to perform the algorithm. All synthetic data used for demonstration is attached in appendix. We will also show that there exists a polynomial time algorithm for DMI-clustering when $d$ is a constant. We first introduce the clustering method.

\begin{algorithm}[H]
\SetAlgoLined
\DontPrintSemicolon
\KwInput{$\mathbf{A}\in\mathbb{R}^{n\times d}$ \tcc*{clustering the rows of $\mathbf{A}$}} 
\KwOutput{k,$\mathbf{C}^*\in\{0,1\}^{n\times k}$, $L$, $\mathbf{D}^*\in\mathbb{R}^{k\times k}$\tcc*{$k$ clusters, map row $i$ to cluster $c$ if $C^*_{ic}=1$}}
Add column $(1,1,\cdots,1)^{\top}$ to $\mathbf{A}$, i.e., $\mathbf{A}^{+1}=[\mathbf{A} \quad \mathbf{1}]$\;
Find the rank $k$ of $\mathbf{A}^{+1}$\;
Pick $k$ linearly independent columns of $\mathbf{A}^{+1}$ to construct matrix $\tilde{\mathbf{A}}$ arbitrarily\;
$L=$ the set of picked column indices \;
Let $\mathcal{C}$ be $\{\mathbf{C}\in\{0,1\}^{n\times k}|\forall i,\sum_{c}C_{ic}= 1\}$\footnote{We can pick other constraint set $\mathcal{C}$ here, for example, we can use $\mathcal{C}^{soft}:=\{\mathbf{C}\in[0,1]^{n\times k}|\forall i,\sum_{c}C_{ic}= 1\}$}\; 
Find $\mathbf{C}^*\in\argmax_{\mathbf{C}\in\mathcal{C}} \det[\mathbf{C}^{\top}\tilde{\mathbf{A}}]$\tcc*{we call $|\det[\mathbf{C}^{\top}\tilde{\mathbf{A}}]|$ the DMI-score of $\mathbf{C}$, $\dmis(\mathbf{C},\mathbf{A})$}
$\mathbf{D}^*=(\mathbf{C}^{*\top} \tilde{\mathbf{A}})^{-1}$
  \caption{\textbf{DMI-clustering} $\dmic(\mathbf{A})$} \label{alg:dmic}
  
 \end{algorithm}
 
The core of DMI-clustering is an evaluation metric for each clustering, DMI-score. We will show that intuitively, DMI-score is the volume of the simplex formed by clusters' means, multiplying the product of cluster sizes. DMI-clustering aims to find the clustering that has the highest DMI-score. A naive attempt to solve the optimization problem in DMI-clustering is to enumerate all possible $\mathbf{C}$. However, it needs an exponential running time in $n$, even when $d$ is a constant. To have a better algorithm, in addition to show that DMI-clustering is affine-invariant, we also analyze the clustering result of DMI-clustering and find that it is linearly-partitioned. This linearly-partitioned property will directly induce a polynomial time algorithm when $d$ is a constant and a simple local maximal-searching heuristics in general. We will use the following notation to describe the clustering result of DMI-clustering. 

\begin{definition}[Index of row's max]
	For matrix $\mathbf{M}\in\mathbb{R}^{n\times k}$, we define $\idxmax(\mathbf{M})\in\mathbb{R}^{n\times k}$ as a matrix such that for all $i$, $\idxmax(\mathbf{M})$'s $i^{th}$ row is a one-hot vector that indicates the index of the largest entry of $\mathbf{M}$'s $i^{th}$ row. When there is a tie, we define $\idxmax(\mathbf{M})$ as a set of all matrices whose rows indicate the index of the largest entry of $\mathbf{M}$'s rows. For example, \[\idxmax\left({\begin{bmatrix}
	.3 & .4 & .1\\
	0 & -1 & .1
\end{bmatrix}}\right)=\begin{bmatrix}
	0 & 1 & 0\\
	0 & 0 & 1
\end{bmatrix}.\]

\end{definition}

 We start to introduce the theoretic properties of DMI-clustering and visualize them in Figure~\ref{fig:prop}. 
\begin{figure}
	\includegraphics[width=17cm]{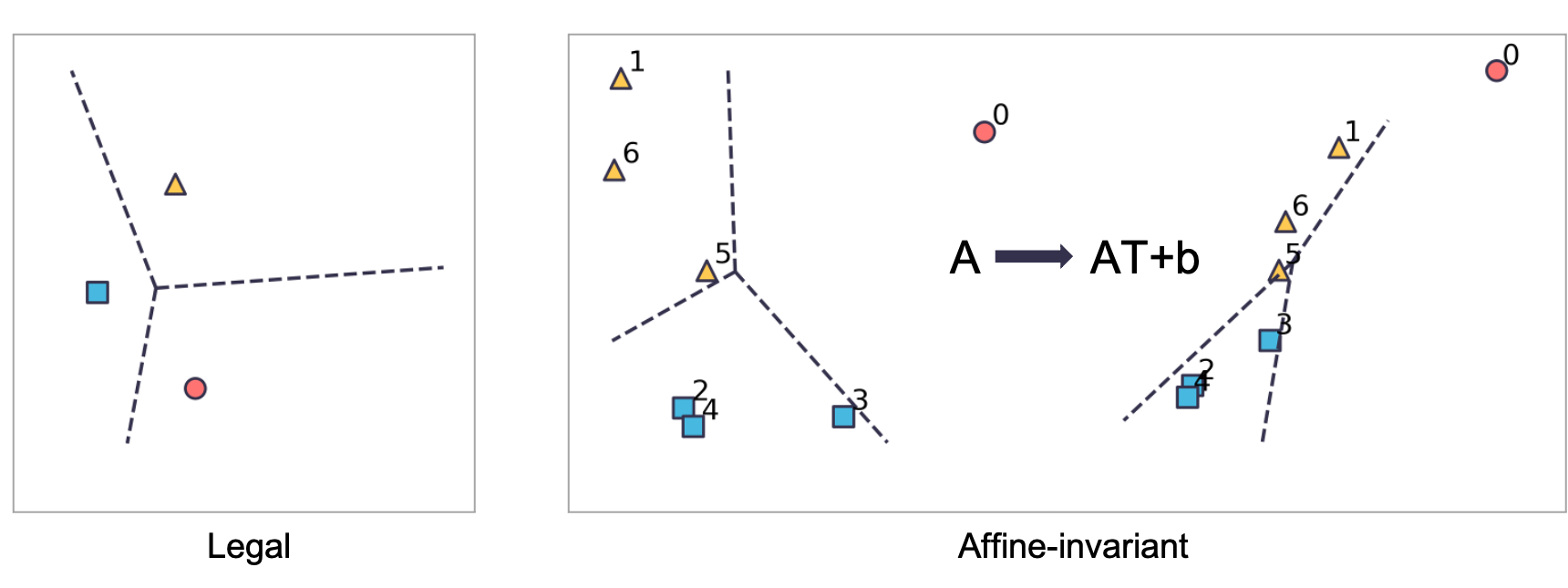}
	\caption{Geometric interpretation of DMI-clustering's properties: 1) legal: when the set of points are $k\leq d+1$ distinct points (e.g. $k=3,d=2$, $\mathbf{A}^{\top}=[\mathbf{a}_1^{\top},\mathbf{a}_1^{\top},\mathbf{a}_2^{\top},\mathbf{a}_3^{\top},\mathbf{a}_2^{\top}]$), each of the distinct points will form a cluster; 2) affine-invariant: when we perform the same affine-transformation to all points, the clustering result will not change; 3) Linearly-partitioned: for all $i\in[n]$, if row $i$ is mapped to cluster $c$, then the row vector $\tilde{\mathbf{A}}_i$ has the highest inner product to column vector $\mathbf{D}_c$ among all column vectors of $\mathbf{D}^*$. Thus, $\mathbf{D}^*$ determines a partition (the black dashed line) and the mean of all points lies in the intersection of the partition lines. \label{fig:prop}
  }
\end{figure}


\begin{theorem}\label{thm:main}
	DMI-clustering is 	
	\begin{description}
		\item [Legal] when there exists full rank $\mathbf{C}^*\in\mathcal{C}$ such that $\mathbf{A}=\mathbf{C}^{*\top}\mathbf{B}$ where $\mathbf{B}$ is invertible, for both hard constraint $\mathcal{C}=\{\mathbf{C}\in\{0,1\}^{n\times k}|\forall i,\sum_{c}C_{ic}= 1\}$ and soft constraint $\mathcal{C}^{soft}$, $\dmic(\mathbf{A})$ equals $\mathbf{C}^*$ (up to column permutation).  
		\item [Affine-invariant] the DMI-score of $\mathbf{C}$ is relatively invariant to any invertible affine transformation of the data points, i.e., for all invertible $\mathbf{T}\in\mathbb{R}^{d\times d}$ and $\mathbf{b}\in\mathbb{R}^{1\times b}$, for all $\mathbf{C}_1,\mathbf{C}_2$ which have non-zero DMI-score regarding $\mathbf{A}$, 
		\[\frac{\dmis(\mathbf{C}_1,\mathbf{A})}{\dmis(\mathbf{C}_2,\mathbf{A})}=\frac{\dmis(\mathbf{C}_1,\mathbf{A}\mathbf{T}+\mathbf{b})}{\dmis(\mathbf{C}_2,\mathbf{A}\mathbf{T}+\mathbf{b})}\footnote{When we add a row vector to a matrix, we add the vector to all rows of the matrix.} \] thus, for any constraint set $\mathcal{C}$, $\dmic(\mathbf{A}\mathbf{T}+\mathbf{b})$ is equivalent to $\dmic(\mathbf{A})$ in the sense that any solution of $\dmic(\mathbf{A}\mathbf{T}+\mathbf{b})$ is also a solution of $\dmic(\mathbf{A})$, vice versa;
		\item [Linearly-partitioned] when we maximize over $\mathcal{C}=\{\mathbf{C}\in\{0,1\}^{n\times k}|\forall i,\sum_{c}C_{ic}= 1\}$, $\mathbf{C}^*$ is a local maximal of $\det[\mathbf{C}^{\top}\tilde{\mathbf{A}}]$ if and only if $\mathbf{C}^*\in \idxmax(\tilde{\mathbf{A}}\mathbf{D}^*)$ where $\mathbf{D}^*=(\mathbf{C}^{*\top} \tilde{\mathbf{A}})^{-1}$.
	\end{description}
	Geometrically, the average of $\tilde{\mathbf{A}}$'s row vectors is the intersection of all partitions, i.e., $\bar{\tilde{\mathbf{a}}} *\mathbf{D}^*\propto \mathbf{1}$ and if $\mathbf{A}-\bar{\mathbf{a}}$'s rank is d, the optimization goal is equivalent with maximizing the volume of the simplex consist of each cluster's mean, multiplying the product of the cluster sizes. 	
\end{theorem}

For simple one-dimensional case, we show that \dmi-clustering will partition the numbers based on the mean (Corollary~\ref{coro:1d}). For other dimensions, the linearly-partitioned property directly induce an algorithm (Algorithm~\ref{alg:kco}) to find the local optimum. This algorithm is very similar to Lloyd's algorithm for k-means: 1. assignment: assign the points according to the current partition; 2. update: update the partition using the current clusters. We call this algorithm k-cofactors since the update step employs the cofactors of the matrix which consists of the current cluster means. In the 2d case, we will translate the operations in k-cofactors into natural geometric operations (Algorithm~\ref{alg:2d}, Figure~\ref{fig:kco2d}). In the end of this section, we will also show that the linearly-partitioned property induces a polynomial algorithm when the dimension $d$ is a constant, based on properties of shatter functions. 

\begin{corollary}[1d points]

	For n real numbers $a_1,a_2,\cdots,a_n$ where there exists $i\neq j$ such that $a_i\neq a_j$, \dmi-clustering will partition the numbers into two clusters: above the average and below the average.
	\end{corollary}\label{coro:1d}
\begin{proof}[Proof of Corollary~\ref{coro:1d}]
	Since \dmi-clustering is affine-invariant, we subtract the average from each number and define vector $\mathbf{v}:=(a_1-\bar{a},a_2-\bar{a},\cdots,a_n-\bar{a})$. 
	
	\begin{align*}
		\det[[\mathbf{p},\mathbf{1}-\mathbf{p}]^{\top}[\mathbf{v}\quad \mathbf{1}]]&=  \det[\begin{bmatrix}\mathbf{p}^{\top} \mathbf{v}& \mathbf{p}^{\top}\mathbf{1} \\(\mathbf{1}-\mathbf{p})^{\top}\mathbf{v} & (\mathbf{1}-\mathbf{p})^{\top}\mathbf{1}\\ \end{bmatrix}]\\ \tag{add first row the second row}
		& \propto \det[\begin{bmatrix}\mathbf{p}^{\top} \mathbf{v}& \mathbf{p}^{\top}\mathbf{1} \\\mathbf{1}^{\top}\mathbf{v} & \mathbf{1}^{\top}\mathbf{1}\\ \end{bmatrix}]\\ \tag{the sum of $\mathbf{v}$ is 0}
		& = \det[\begin{bmatrix}\mathbf{p}^{\top} \mathbf{v}& \mathbf{p}^{\top}\mathbf{1} \\0 & \mathbf{1}^{\top}\mathbf{1}\\ \end{bmatrix}]\\
		& \propto \mathbf{p}^{\top} \mathbf{v}
	\end{align*}
	
	To maximize $\mathbf{p}^{\top} \mathbf{v}$, the 0-1 vector $\mathbf{p}$ will pick all positive entries of $\mathbf{v}$. Thus, \dmi-clustering will partition the numbers into two clusters: above the average and below the average.
\end{proof}

 \begin{algorithm}[H]
\SetAlgoLined
\DontPrintSemicolon
\KwInput{$\tilde{\mathbf{A}}\in\mathbb{R}^{n\times k}$ } 
\KwOutput{$\mathbf{C}\in\{0,1\}^{n\times k}$}
Initialize $\mathbf{C}$\;
$\mathbf{D}:=(\mathbf{C}^{\top} \tilde{\mathbf{A}})^{-1}$\;
\While{$\mathbf{C}\notin \idxmax(\tilde{\mathbf{A}}\mathbf{D})$}
	{
	$\mathbf{C}\leftarrow \idxmax(\tilde{\mathbf{A}}\mathbf{D})$\footnote{Update $\mathbf{C}$'s row only if it does not indicate the index of the highest entry of $\tilde{\mathbf{A}}\mathbf{D}$'s row} \tcc*{assignment step}
	$\mathbf{D}:=(\mathbf{C}^{\top} \tilde{\mathbf{A}})^{-1}$ \tcc*{update step}
	}
\Return $\mathbf{C}$
  \caption{\textbf{k-cofactors}: finding local maximal of $\max_{\mathbf{C}\in\mathcal{C}} \det[\mathbf{C}^{\top}\tilde{\mathbf{A}}]$}\label{alg:kco}
\end{algorithm} 

\begin{figure}
	\includegraphics[width=15cm]{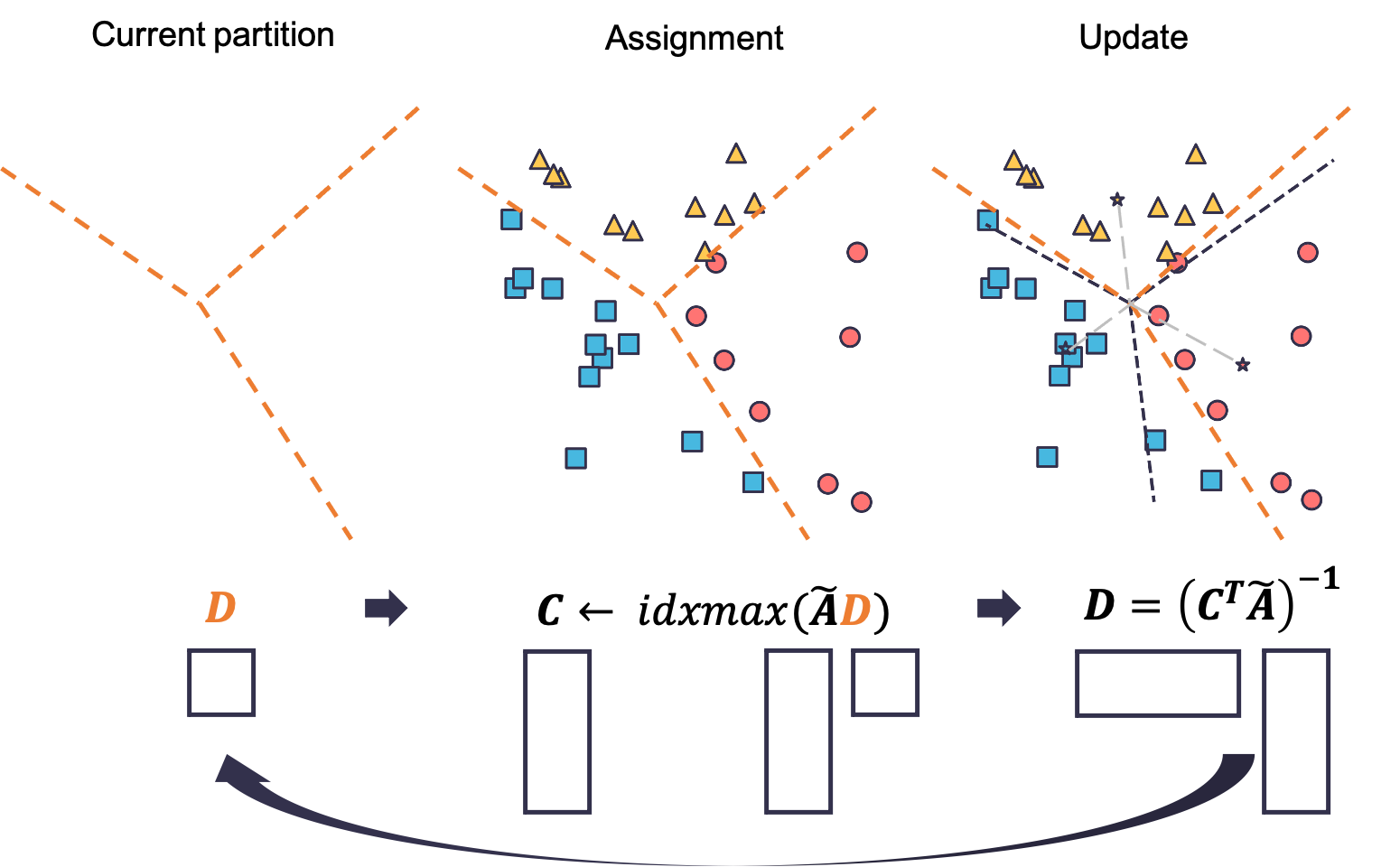}
	\caption{Demonstration of k-cofactors algorithm: the above figure demonstrates a running example for 2d points and shows the repeated steps during the process: left: the current $\mathbf{D}$ determines a partition (the orange dashed lines); middle (assignment step): points are clustered by the partition; right (update step): a new $\mathbf{D}$ is generated by the current clusters.  }\label{fig:kco2d}
\end{figure}

\begin{figure}
	\includegraphics[width=15cm]{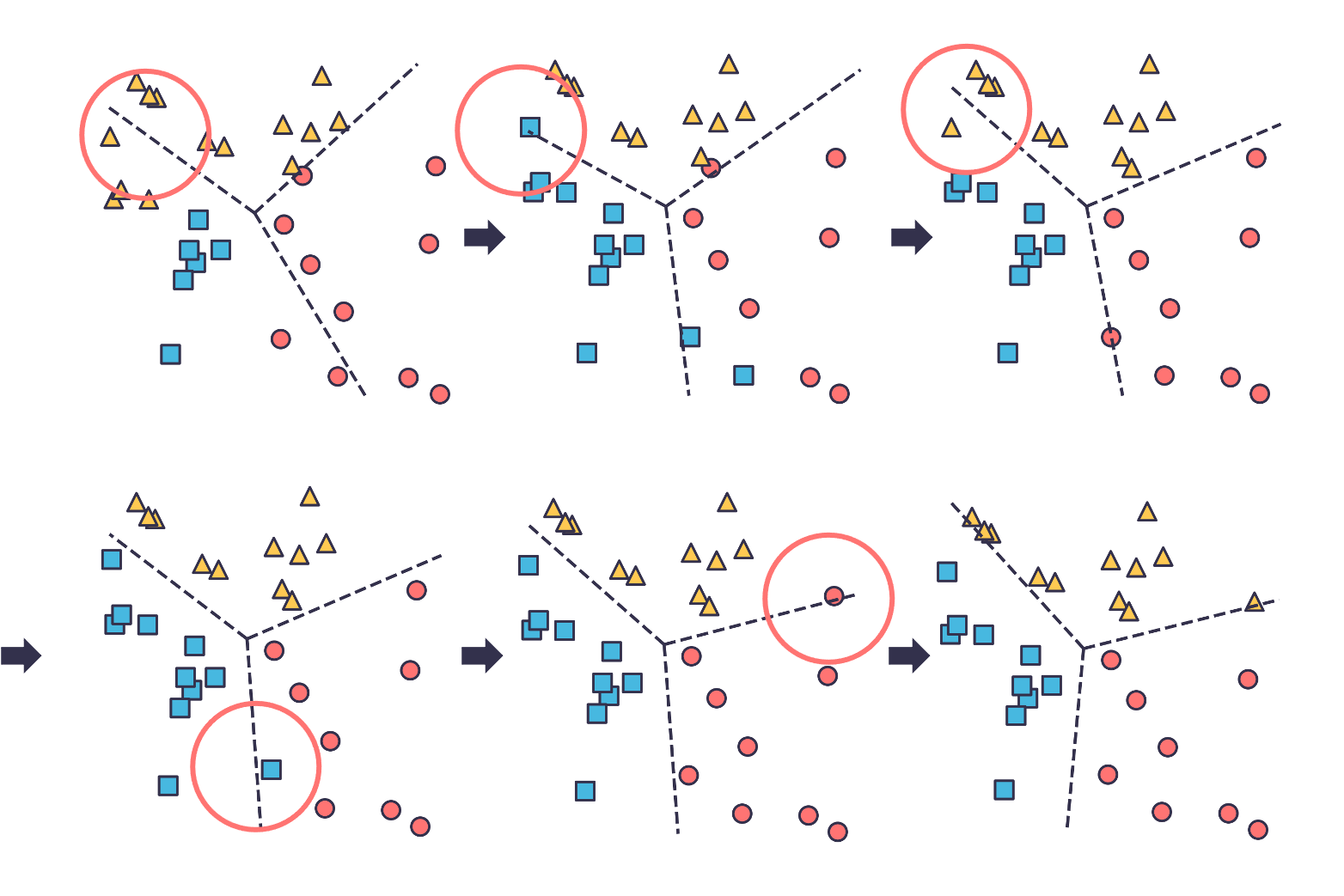}
	\caption{An empirical run of k-cofactors algorithm: the initial k clusters are created by $\mathbf{C}\leftarrow\idxmax(\tilde{\mathbf{A}}-\bar{\tilde{\mathbf{a}}})$. Each subfigure shows the current cluster $\mathbf{C}$ and its generated partition $\mathbf{D}=(\mathbf{C}^{\top}\tilde{A})^{-1}$ (the black dashed lines). Note that the process stops as long as the generated partition $\mathbf{D}=(\mathbf{C}^{\top}\tilde{A})^{-1}$ is consistent with the current cluster, i.e., $\mathbf{C}\notin \idxmax(\tilde{\mathbf{A}}\mathbf{D})$. The red circle marks the inconsistent part. }\label{fig:empirical}
\end{figure}

\begin{algorithm}[H]
\SetAlgoLined
\DontPrintSemicolon
\KwInput{$(\mathbf{a}_1,\mathbf{a}_2,\cdots,\mathbf{a}_n)$ \tcc*{$n$ 2d points}} 
\KwOutput{$\mathbf{C}\in\{0,1\}^{n\times 3}$ \tcc*{$3$ clusters, $\mathbf{a}_i$ belongs to cluster $c$ if $C_{ic}=1$}}
Initialize $\mathbf{C}$\;
Calculate mean $\bar{\mathbf{a}}$ of all points\tcc*{The center intersection point in Figure~\ref{fig:kco2d}}
Calculate mean $\bar{\mathbf{a}}_c$ of each cluster $c=1,2,3$ \tcc*{The stars in Figure~\ref{fig:kco2d}}
Draw partition line, starting from $\bar{\mathbf{a}}$, by extending line segment $[\bar{\mathbf{a}}_c,\bar{\mathbf{a}}]$ for all $c$ \tcc*{the black dashed lines in Figure~\ref{fig:kco2d}}
\While{the partition is not consistent with the current cluster}
	{
	Assign points to clusters based on the current partition\tcc*{assignment step}
	Update the partition based on the current clusters \tcc*{update step}
	}
\Return $\mathbf{C}$
  \caption{k-cofactors for 2d points }\label{alg:2d}
\end{algorithm}

\begin{corollary}[2d points]\label{coro:2d}
	For any $\mathbf{A}\in\mathbb{R}^{n\times 2}$, if $\mathbf{A}-\bar{\mathbf{a}}$'s rank is two\footnote{Otherwise it will reduce to the 1d case.}, then running $\dmic(\mathbf{A})$ (Algorithm~\ref{alg:dmic}) with k-cofactors algorithm (Algorithm~\ref{alg:kco}) is equivalent to Algorithm~\ref{alg:2d}. Geometrically, the optimization goal is equivalent with maximizing the area of the triangles $\Delta_{\bar{\mathbf{a}}_1,\bar{\mathbf{a}}_2,\bar{\mathbf{a}}_3}$ multiplying the $n_1 n_2 n_3$ where $n_c$ is the number of points in cluster $c$. 
\end{corollary}


\begin{proof}[Proof of Corollary~\ref{coro:2d}]

If $\mathbf{A}-\bar{\mathbf{a}}$'s rank is two, then $[\mathbf{A}\quad \mathbf{1}]$ is rank three. Therefore, we will pick $\tilde{\mathbf{A}}=[\mathbf{A}\quad \mathbf{1}]$.

$\mathbf{C}^{\top} \tilde{\mathbf{A}}=\diag(n_1,n_2,n_3)\begin{bmatrix}
	&\bar{\mathbf{a}}_1 &\quad 1\\
	&\bar{\mathbf{a}}_2 &\quad 1\\
	&\bar{\mathbf{a}}_3 &\quad 1\\
\end{bmatrix}$ where $n_c$ is the number of points in cluster $c$. Thus maximizing $\det[\mathbf{C}^{\top} \tilde{\mathbf{A}}]$ is equivalent to maximizing the area of the triangles $\Delta_{\bar{\mathbf{a}}_1,\bar{\mathbf{a}}_2,\bar{\mathbf{a}}_3}$ multiplying the $n_1 n_2 n_3$. 

For $\mathbf{D}=(\mathbf{C}^{\top} \tilde{\mathbf{A}})^{-1}=[\mathbf{d}_1,\mathbf{d}_2,\mathbf{d}_3]$, we have 
\[ [\bar{\mathbf{a}}_1 \quad 1]*\mathbf{d}_1=\frac{1}{n_1}\quad [\bar{\mathbf{a}}_1 \quad 1]*\mathbf{d}_2 =[\bar{\mathbf{a}}_1 \quad 1]*\mathbf{d}_3=0\] We also have that 
\[[\bar{\mathbf{a}}\quad \mathbf{1}]*\mathbf{d}_1=(\sum_c \frac{n_c}{n} [\bar{\mathbf{a}}_c\quad \mathbf{1}])*\mathbf{d}_1=\frac1n \]

and by similar analysis, $[\bar{\mathbf{a}}\quad \mathbf{1}]*\mathbf{d}_2=[\bar{\mathbf{a}}\quad \mathbf{1}]*\mathbf{d}_3=\frac1n$

	For point $\mathbf{a}= \lambda\bar{\mathbf{a}}+(1-\lambda) \bar{\mathbf{a}}_1$,
	
	\[ [\mathbf{a} \quad \mathbf{1}]*\mathbf{d}_1=\lambda \frac1n + (1-\lambda)\frac{1}{n_1} \]
	
	and 
	
	\[ [\mathbf{a} \quad \mathbf{1}]*\mathbf{d}_2=[\mathbf{a} \quad \mathbf{1}]*\mathbf{d}_3=\lambda\frac1n \]
	
	when $\lambda< 1$, $\mathbf{a}$ should be mapped to cluster 1 and when $\lambda\geq 1$, $\mathbf{a}$ can be mapped to either cluster 2 or cluster 3. Thus, we can obtain the partition line for cluster 2 and 3, call it line 2-3, by extending line segment $[\bar{\mathbf{a}}_1,\bar{\mathbf{a}}]$. By analogous analysis for $\bar{\mathbf{a}}_2, \bar{\mathbf{a}}_3$, we obtain three partition lines like Figure~\ref{fig:kco2d}. 
	
	For each point $\mathbf{x}$ between line 1-2 and 2-3, we can represent it as $\mathbf{x}=\bar{\mathbf{a}}+s (\bar{\mathbf{a}}-\bar{\mathbf{a}}_3)+ t  (\bar{\mathbf{a}}-\bar{\mathbf{a}}_1)$ where $s,t>0$. Then $[\mathbf{x}\quad \mathbf{1}]*\mathbf{d}_2=\frac1n (1+s+t)$, $[\mathbf{x}\quad \mathbf{1}]*\mathbf{d}_1=\frac1n (1+s+t)-t\frac{1}{n_1}$ and $[\mathbf{x}\quad \mathbf{1}]*\mathbf{d}_3=\frac1n (1+s+t)-s\frac{1}{n_3}$. Thus, $\mathbf{x}$ belongs to cluster 2. By analogous analysis for other points, we finish the proof.

	\end{proof}

\begin{proof}[Proof of Theorem~\ref{thm:main}]
	We first prove the affine-invariance. Note that \[[\mathbf{A}\mathbf{T}+\mathbf{b}\quad \mathbf{1}]=[\mathbf{A}\quad \mathbf{1}]\begin{bmatrix}\mathbf{T}&0\\\mathbf{b}&1\end{bmatrix}.\]
	
	Thus, for invertible $\mathbf{T}$, both $[\mathbf{A}\mathbf{T}+\mathbf{b}\quad \mathbf{1}]$ and $[\mathbf{A}\quad \mathbf{1}]$ have the same rank $k$.
	
	Since $\tilde{\mathbf{A}}$ are $k$ linearly independent columns of $[\mathbf{A}\quad \mathbf{1}]$, there exists $\mathbf{P}\in\mathbb{R}^{k\times (d+1)}$ such that $[\mathbf{A}\quad \mathbf{1}]=\tilde{\mathbf{A}}\mathbf{P}$. Therefore, we have $[\mathbf{A}\mathbf{T}+\mathbf{b}\quad \mathbf{1}]=\tilde{\mathbf{A}}\mathbf{P}\begin{bmatrix}\mathbf{T}&0\\\mathbf{b}&1\end{bmatrix}$. 	
	
	Thus, for any $\tilde{\mathbf{A}}'$ that consists of $k$ linearly independent columns in $[\mathbf{A}\mathbf{T}+\mathbf{b}\quad \mathbf{1}]$, we pick corresponded column vectors of $\mathbf{P}\begin{bmatrix}\mathbf{T}&0\\\mathbf{b}&1\end{bmatrix}$ to construct matrix $\mathbf{Q}\in\mathbb{R}^{k\times k}$. We have $\tilde{\mathbf{A}}'=\tilde{\mathbf{A}}\mathbf{Q}$ and $\mathbf{Q}$ is invertible (otherwise $\tilde{\mathbf{A}}'$'s rank will be less than $k$ which is a contradiction). 
	
	Therefore, \[\det[\mathbf{C}^{\top}\tilde{\mathbf{A}'}]=\det[\mathbf{C}^{\top}\tilde{\mathbf{A}}]\det[\mathbf{Q}]\propto\pm\det[\mathbf{C}^{\top}\tilde{\mathbf{A}}]
\] due to the multiplicative property of determinant. The relative-invariance of DMI-score and affine-invariance of DMI-clustering directly follows from the above formula. 

We then show that DMI-clustering is legal. Note that since it is affine-invariant, we only need to analyze $\dmic(\mathbf{C}^*)$. Since the sum of the column vectors of $\mathbf{C}^*$ is all one vector $\mathbf{1}$ and the rank of $\mathbf{C}^{\top}$ is $d$, we can pick $\mathbf{C}^*$ itself as the matrix that consists of $d$ linearly independent columns of $[\mathbf{C}^*\quad\mathbf{1}]$. 

We normalize $\mathbf{C}^*$ by dividing each column vector by its sum to obtain $\ncol(\mathbf{C}^*)$. Note that this normalization is equivalent to multiplying a diagonal matrix on the right. Thus, maximizing $\det[\mathbf{C}^{\top}\mathbf{C}^*]$ over $\mathbf{C}\in\mathcal{C}^{soft}$ is equivalent to $\max_{\mathbf{C}\in\mathcal{C}^{soft}}\det[\mathbf{C}^{\top}\ncol(\mathbf{C}^*)]$. 

Moreover, $\mathbf{C}^{*\top}\ncol(\mathbf{C}^*)$ is an identity matrix whose determinant is one. Since both $\mathbf{C}^{\top}$ and $\ncol(\mathbf{C}^*)$ are column-stochastic matrices\footnote{Each entry is in [0,1] and each column sums to 1.}, their product is still a column-stochastic matrix. When $\mathbf{C}$ does not equal to $\mathbf{C}^*$ even up to column permutation, i.e., there does not exist a permutation matrix $\pi$ such that $\mathbf{C}=\pi\mathbf{C}^*$, the product of $\mathbf{C}^{\top}$ and $\ncol(\mathbf{C}^*)$ will be a non-permutation column-stochastic matrix whose determinant is strictly less than 1. Therefore, when there exists full rank $\mathbf{C}^*\in\mathcal{C}$ such that $\mathbf{A}=\mathbf{C}^{*\top}\mathbf{B}$ where $\mathbf{B}$ is invertible, $\dmic(\mathbf{A})$ equals $\mathbf{C}^*$ (up to column permutation).  

Finally, we prove that the clustering result is linearly-partitioned. $\tilde{\mathbf{A}}$'s row vectors are $\tilde{\mathbf{a}}_1,\tilde{\mathbf{a}}_2,\cdots,\tilde{\mathbf{a}}_n$. $\mathbf{C}^{*\top} \tilde{\mathbf{A}}$'s row vectors are denoted by $\mathbf{v}_1,\mathbf{v}_2,\cdots,\mathbf{v}_k$. 

For all $i\in[n]$, we denote the cluster it is mapped to by $c_i^*$, i.e., $\mathbf{C}^*_{i,c_i^*}=1$. For all $c\in[k]$, when $\mathbf{C}^*$ is a local maximum, the current assignment is better than moving $i$ to cluster $c$, that is,

\[\det[\mathbf{C}^{*\top} \tilde{\mathbf{A}}]=\det[\begin{bmatrix}\mathbf{v}_1\\\mathbf{v}_2\\\vdots\\\mathbf{v}_k\end{bmatrix}]\geq \det[\begin{bmatrix}\vdots\\\mathbf{v}_{c_i^*}-\tilde{\mathbf{a}}_i\\ \vdots\\\mathbf{v}_c+\tilde{\mathbf{a}}_i\\\vdots\end{bmatrix}]\]

Due to the multilinear property of determinant, 

\begin{align*}
	\det[\begin{bmatrix}\vdots\\\mathbf{v}_{c_i^*}-\tilde{\mathbf{a}}_i\\ \vdots\\\mathbf{v}_c+\tilde{\mathbf{a}}_i\\\vdots\end{bmatrix}]=&\det[\begin{bmatrix}\vdots\\\mathbf{v}_{c_i^*}\\ \vdots\\\mathbf{v}_c+\tilde{\mathbf{a}}_i\\\vdots\end{bmatrix}]-\det[\begin{bmatrix}\vdots\\\tilde{\mathbf{a}}_i\\ \vdots\\\mathbf{v}_c+\tilde{\mathbf{a}}_i\\\vdots\end{bmatrix}]\\
	=&\det[\begin{bmatrix}\vdots\\\mathbf{v}_{c_i^*}\\ \vdots\\\mathbf{v}_c\\\vdots\end{bmatrix}]+\det[\begin{bmatrix}\vdots\\\mathbf{v}_{c_i^*}\\ \vdots\\\tilde{\mathbf{a}}_i\\\vdots\end{bmatrix}]-\det[\begin{bmatrix}\vdots\\\tilde{\mathbf{a}}_i\\ \vdots\\\mathbf{v}_c\\\vdots\end{bmatrix}]
\end{align*}
Thus, we have \[\det[\begin{bmatrix}\vdots\\\tilde{\mathbf{a}}_i\\ \vdots\\\mathbf{v}_c\\\vdots\end{bmatrix}]\geq \det[\begin{bmatrix}\vdots\\\mathbf{v}_{c_i^*}\\ \vdots\\\tilde{\mathbf{a}}_i\\\vdots\end{bmatrix}]\] for all $c$.

Since $\mathbf{D}^*:=(\mathbf{C}^{*\top} \tilde{\mathbf{A}})^{-1}$, $\det[\mathbf{C}^{*\top}\tilde{\mathbf{A}}]D^*_{ij}$ is the cofactors of matrix $\mathbf{C}^{*\top}\tilde{\mathbf{A}}$. We denote $\mathbf{D}^*$'s column vectors by $\mathbf{d}^*_1,\mathbf{d}^*_2,\cdots,\mathbf{d}^*_k$. 

Due to Laplace expansion, we have 

\[ \tilde{\mathbf{a}}_i*(\det[\mathbf{C}^{*\top}\tilde{\mathbf{A}}]\mathbf{d}^*_{c_i^*})=\det[\begin{bmatrix}\vdots\\\tilde{\mathbf{a}}_i\\ \vdots\\\mathbf{v}_c\\\vdots\end{bmatrix}]\geq\det[\begin{bmatrix}\vdots\\\mathbf{v}_{c_i^*}\\ \vdots\\\tilde{\mathbf{a}}_i\\\vdots\end{bmatrix}]= \tilde{\mathbf{a}}_i*(\det[\mathbf{C}^{*\top}\tilde{\mathbf{A}}]*\mathbf{d}^*_{c}) \]

Since $\det[\mathbf{C}^{*\top}\tilde{\mathbf{A}}]$ is the maximum that must be positive, we have for all $i$, for all $c$, 

\[ \tilde{\mathbf{a}}_i*\mathbf{d}^*_{c_i^*}\geq \tilde{\mathbf{a}}_i*\mathbf{d}^*_{c}. \] 

Thus, $\mathbf{C}^*\in \idxmax(\tilde{\mathbf{A}}\mathbf{D}^*)$. Moreover,

\begin{align*}
	\mathbf{1}\mathbf{D}^{*-1}=\mathbf{1}\mathbf{C}^{*\top}\tilde{\mathbf{A}} = \mathbf{1} \tilde{\mathbf{A}} \propto \bar{\tilde{\mathbf{a}}}
\end{align*}

which shows that $\bar{\tilde{\mathbf{a}}} *\mathbf{D}^*\propto \mathbf{1}$, i.e., the mean of all data points is in the intersection of all partitions. We also have

\[\mathbf{C}^{\top} \tilde{\mathbf{A}}=\diag(n_1,n_2,\cdots,n_k)\begin{bmatrix}
	&\bar{\mathbf{a}}_1 &\quad 1\\
	&\bar{\mathbf{a}}_2 &\quad 1\\
	&\vdots &\quad \vdots\\
	&\bar{\mathbf{a}}_k &\quad 1\\
\end{bmatrix}\]

whose determinant is proportional to the volume of the simplex formed by $(\bar{\mathbf{a}}_1,\bar{\mathbf{a}}_2,\cdots,\bar{\mathbf{a}}_k)$, multiplying the product of cluster sizes $n_1,n_2,\cdots,n_k$. We finish the proof.

\end{proof}

The linearly-partitioned part in the above theorem shows that instead of trying all possible $n\times k$ $\mathbf{C}$, we can try all possible $k\times k$ $\mathbf{D}$ and check whether its corresponding $\mathbf{C}$ is optimal to obtain the result of DMI-clustering. In particular, when dimension $d$ is a constant, there exists the number of distinct linear-partitions are polynomial in the number of data points $n$ due to theory of VC dimension \cite{blumer1989learnability}. This implies the following corollary. 
\begin{corollary}
	 When dimension $d$ is a constant, when we maximize over the default constraint $\mathcal{C}=\{\mathbf{C}\in\{0,1\}^{n\times k}|\forall i,\sum_{c}C_{ic}= 1\}$, there exists a polynomial algorithm for DMI-clustering. 
\end{corollary}

\begin{proof}

When $d$ is a constant, the number of clusters $k\leq d+1$ is a constant as well. Given $n$ data points, for any linear-partition $\mathbf{D}^*$, the points that belong to cluster 1 is the set of all points $\{\tilde{a}|\tilde{a}*\mathbf{d}^*_1>\tilde{a}*\mathbf{d}^*_c,c\neq 1\}$ thus is intersections of $k-1$ half-spaces. Since half-space's VC dimension is a constant $k+1$ here, there will be at most polynomial $O(n^{(k-1)(k+1)})$ distinct cluster 1. Moreover, the number of clusters $k$ is a constant here. Thus there will be most $O(n^{k(k-1)(k+1)})$ distinct linear-partitions. Thus, we can enumerate these linear-partitions in polynomial time and check which one gives the best cluster $\mathbf{C}^*$ that maximizes the determinant. 
	
\end{proof}
  
In the next section, we will apply DMI-clustering to extract knowledge from elicited information without ground truth. We will show that this method will be robust to people's strategies and expertise. 


\section{Multi-task \& Invariant Clustering}

This section will introduce a multi-task information elicitation mechanism which is not only dominantly-truthful but also extracts knowledge from the collected answers. The extraction is robust when each task is answered by a large number of randomly selected people.


\subsection{Knowledge-Extracted DMI (K-DMI)-Mechanism}

\paragraph{Setting} There are $n$ a priori similar tasks and $m$ agents. Each task is a multi-choice question with $C$ options. Each agent is assigned a random set of tasks (size $\geq 2C$) in random order independently. Agents finish tasks without communication. Each agent $i$'s reports are represented as a $n\times C$ 0-1 report matrix $\mathbf{R}_i$ where the $(t,c)$ entry of $\mathbf{R}_i$ is 1 if and only if agent $i$ reports option $c$ for task $t$. The all zero rows in $\mathbf{R}_i$ are the tasks that agent $i$ does not perform.

\paragraph{Model} Tasks have states which are drawn i.i.d. from an unknown distribution. Each task $t$'s state is a vector $\mathbf{a}_t\in[0,1]^C$ where each agent will receive signal $c$ with probability $\mathbf{a}_t(c)$. Each agent $i$ is associated with a strategy $\mathbf{S}_i\in[0,1]^{C\times C}$ where $\mathbf{S}_i(c,c')$ is the probability that she receives $c$ and reports $c'$. We assume that each agent performs the same strategy for all tasks. This assumption is reasonable since the tasks are a priori similar and agents receive tasks in a random order independently. A payment scheme is dominantly truthful if for each agent, regardless of other people's strategies, telling the truth pays her higher\footnote{not strictly higher} than any other strategy in expectation. 


\begin{algorithm}[H]
\SetAlgoLined
\DontPrintSemicolon
\KwInput{$\mathbf{R}_i\in\{0,1\}^{n\times C},i\in[m]$} 
\KwOutput{extracted knowledge $\mathbf{C}^*$,payment $p_i,i\in[m]$}
\tcc{Knowledge extraction}
$\mathbf{A}=\sum_i \mathbf{R}_i $\;
$\forall t,c, A_{tc}=\frac{A_{tc}}{\sum_c A_{tc}}$\;
Get $\mathbf{C}^*$ from $\dmic(\mathbf{A})$\;
\tcc{Payment calculation}
\For {$i\in[m]$}
{
$\mathbf{A}=\sum_{j\neq i} \mathbf{R}_j $\;
$\forall t,c, A_{tc}=\frac{A_{tc}}{\sum_c A_{tc}}$\;
$T_{i}$ is the set of tasks agent $i$ performed, $T_{-i}=[n]/T_{i}$\;

\tcc{Use $T_{-i}$ to obtain the partition to guarantee independence}
Get $\mathbf{D}^*_{-i}$ and $L$ from $\dmic(\mathbf{A}_{T_{-i},:})$\;
\tcc{Use obtained partition to extract knowledge for tasks in $T_i$}
$\mathbf{C}^*_i\leftarrow \idxmax(\mathbf{A}_{T_i,L}\mathbf{D}^*_{-i})$\;
Divide $T_i$ into two disjoint parts $T_1,T_2$ such that both $|T_1|,|T_2|\geq C$\;
Restrict $\mathbf{C}^*_i$ and $\mathbf{R}_i$ to $T_{\ell}$ to get $\mathbf{C}^*_{\ell}$ and $\mathbf{R}_i^{\ell}$ respectively for $\ell=1,2$\;
\tcc{Reward agent $i$ using the extracted knowledge}
$p_i= \det(\mathbf{C}_1^{*\top}\mathbf{R}^1_i)\det(\mathbf{C}_2^{*\top}\mathbf{R}^2_i)$
}\caption{\textbf{Knowledge-Extracted DMI (K-DMI)-Mechanism}} \label{alg:kdmi}
\end{algorithm}

\begin{figure}
	\includegraphics[width=15cm]{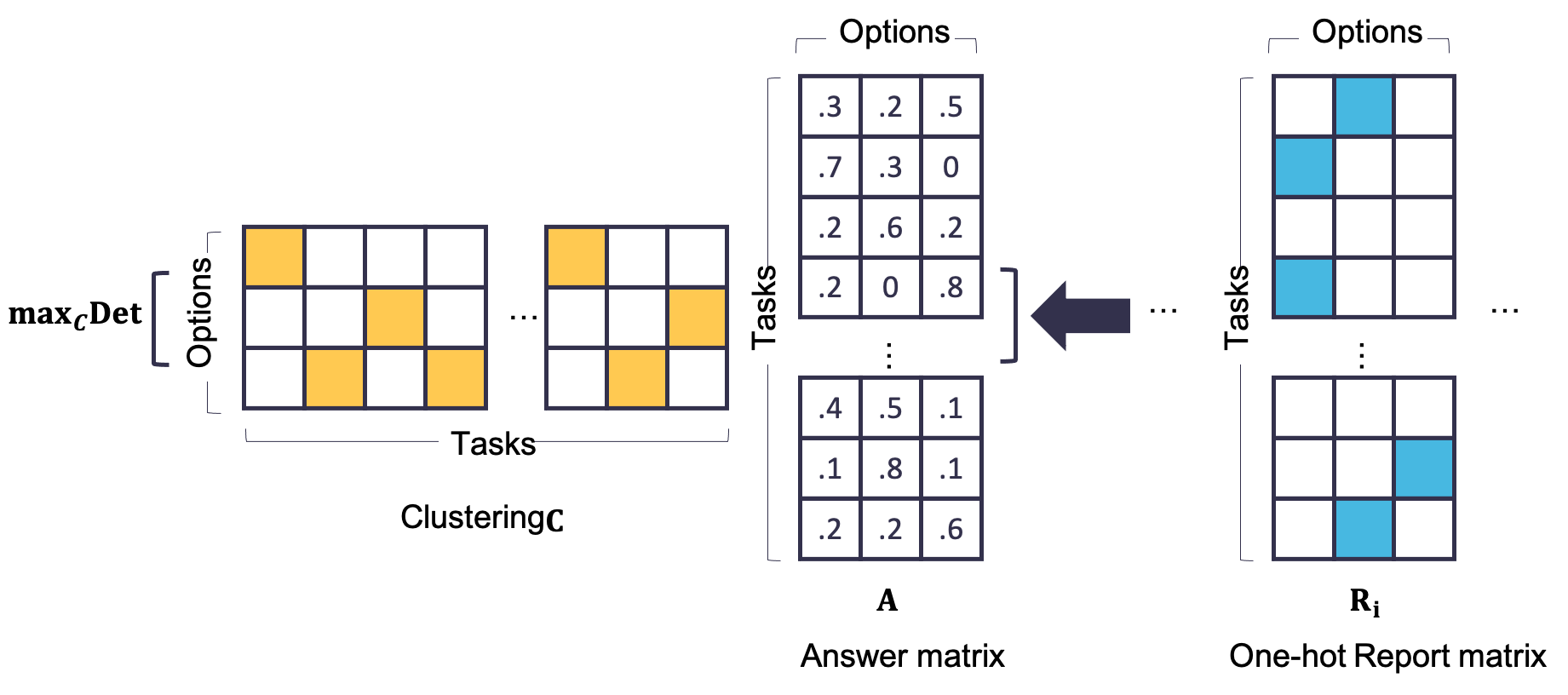}
	\caption{\textbf{Knowledge extraction}: we aggregate all agents' answers and normalize each row to obtain answer matrix $\mathbf{A}$. When $\mathbf{A}$'s rank is $C$, $\dmic(\mathbf{A})$ is equivalent to $\max_{\mathbf{C}}\det[\mathbf{C}^{\top}\mathbf{A}]$.  }
\end{figure}

\begin{figure}
	\includegraphics[width=15cm]{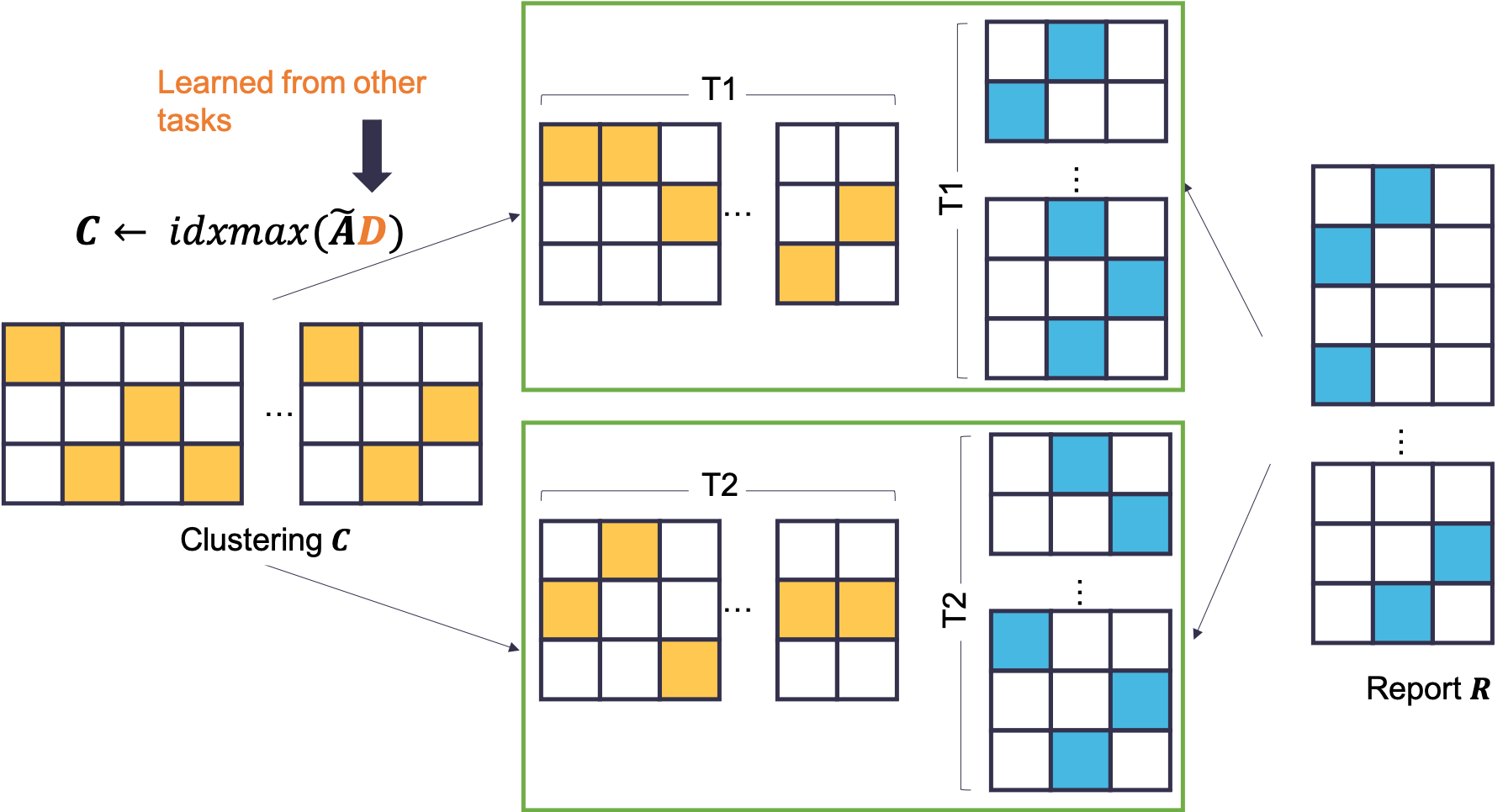}
	\caption{\textbf{Payment calculation}: for each agent, we use the tasks she does not perform to learn the partition $\mathbf{D}$ and induce a clustering result for the tasks she performs. To pay her, we divide the tasks she performs into two disjoint parts $T_1,T_2$. We calculate determinant of the product of the two matrices in each green rectangle and multiply the determinant. Note that it's same as the original \dmi-Mechanism if we change the clustering result to a peer agent's report. }
\end{figure}

K-DMI-Mechanism is a modification of DMI-Mechanism which pays each agent the determinant mutual information between her report and her peer's report in expectation.

\begin{definition}[Determinant based Mutual Information (DMI) \cite{kong2020dominantly}] For two random variables $X,Y$ whose joint distribution is denoted as a matrix $\mathbf{U}_{X,Y}$,\[\dmi(X;Y):=|\det(\mathbf{U}_{X,Y})|\]\end{definition}

Since DMI satisfies data processing inequality and its square has a polynomial format, we can use $\geq 2C$ tasks to implement it. Compared to DMI, K-DMI-Mechanism adds a clustering process, which can be interpreted as finding a clustering which has the highest determinant mutual information with the collection of all agents' answers. For the payment calculation, K-DMI-Mechanism pays each agent the determinant mutual information between her report and the clustering result to reduce variance. The original DMI-Mechanism's payments may not be fair to the hard-working agents when there exists a lot of low-efforts agents. The current payment calculation method may be more robust to other people's strategies. For example, in the simple case where the answer distributions only have $C$ distinct formats, we can use other tasks to learn the correct partition and each agent is rewarded between her report and the underlying ground-truth in K-DMI.

\begin{theorem}\label{thm:kdmi}
	K-DMI-Mechanism is dominantly-truthful. If the number of agents $m$ goes to infinite and $\lim_{m\rightarrow\infty}\frac1m\sum_{i\in[m]}\mathbf{S}_i$ exists and is invertible, the clustering result is invariant to people's strategies. When $\mathbf{A}$'s rank is $C$, $\dmic(\mathbf{A})$ is equivalent to $\argmax_{\mathbf{C}}\det[\mathbf{C}^{\top}\mathbf{A}]$ and $v:=max_{\mathbf{C}}\det[\mathbf{C}^{\top}\mathbf{A}]$ also indicates the quality of answers: $v$ decreases when agents perform strategies. 
\end{theorem}

\begin{proof}[Proof of Theorem~\ref{thm:kdmi}]
	From agent $i$'s view, each task's elicited answers are drawn i.i.d. from a distribution. Thus, the partition $\mathbf{D}_{-i}^*$ learned from the set of tasks $T_{-i}$ is independent of the set of tasks $T_i$ she performs. Therefore, after operating $\mathbf{D}_{-i}^*$ on $T_i$'s answers (excluding agent $i$'s own answer) to obtain a clustering result, the i.i.d. assumption is still satisfied for the pair of agent $i$'s answer and the clustering result $\mathbf{C}_i^*$. We can think the clustering result $\mathbf{C}_i^*$ is answer from a virtual peer agent. The payment calculation reduced to \dmi-Mechanism. Since \dmi-Mechanism is dominantly-truthful, k-DMI-Mechanism is also dominantly truthful. 
	
 For each task $t$, in expectation, 	the honest answer distributions will be $\mathbf{a}_t$. Note that each agent is assigned a random set of tasks independently. We denote the probability each agent is assigned task $t$ as $p$. In expectation, there will be $pm$ agents who assign task $t$. Moreover, after agents perform strategies, before normalization, the expected answer feedback will be $\sum_{i\in[m]} p\mathbf{a}_t \mathbf{S}_i=\mathbf{a}_t \sum_{i\in[m]}p\mathbf{S}_i$. Thus, when $m$ goes to infinite, the expected answer distribution will be $\mathbf{a}_t \lim_{m\rightarrow \infty}\frac1m\sum_{i\in[m]}\mathbf{S}_i$. Since DMI-clustering is affine-invariant, then clustering result is invariant to people's strategies. 
 
 When $\mathbf{A}$'s rank is $C$, we can pick $\tilde{\mathbf{A}}=\mathbf{A}$ (all one column is the sum of $\mathbf{A}$'s columns thus cannot be added here). Thus, $\dmic(\mathbf{A})$ is equivalent to $\argmax_{\mathbf{C}}\det[\mathbf{C}^{\top}\mathbf{A}]$. 
 
 For $v:=max_{\mathbf{C}}\det[\mathbf{C}^{\top}\mathbf{A}]$, when agents perform strategies, it becomes
 \begin{align*}
 v&=max_{\mathbf{C}}\det[\mathbf{C}^{\top}\mathbf{A}(\lim_{m\rightarrow \infty}\frac1m\sum_{i\in[m]}\mathbf{S}_i)]\\
 &\leq max_{\mathbf{C}}\det[\mathbf{C}^{\top}\mathbf{A}]
 \end{align*}
 
since the absolute value of any transition matrix's determinant is less than one. 
 
 Thus, we finish the proof.

	\end{proof}

\paragraph{Discussion for the practical implementation of K-DMI} In DMI-clustering, we cannot distinguish two clustering if they are equivalent up to a permutation for the options. However, if we have a small set of labeled questions, given an agreement measure, we can pick an option-permutation that has the highest agreement with the labeled questions. In addition to the default constraint, we can also pick other constraint set $\mathcal{C}$ without losing the strategy-invariant property. For example, we can use other aggregation methods' results as candidates and determine the best candidate using DMI-clustering. In the payment calculation part, if we have some assumption about each agent's honest answer and the obtained clustering like the determinant of their multiplication is positive in expectation, we do not need to divide the tasks into two parts. We can just pay the agent the determinant of the product of her report matrix and clustering result. In running DMI-clustering, we can employ some easy heuristics to obtain an initialization. For example, we can run the surprisingly popular method or use plurality vote of high-payment agents to get an initialization.

\subsection{DMI-clustering vs. Surprisingly Popular}

Surprisingly popular method picks the option which is the most surprisingly popular compare to the prior. In the multi-task setting, we can use the average to estimate the prior. 

\begin{definition}[Surprisingly popular]
	 For task whose answer distribution is $\mathbf{a}$, we map it to cluster $\argmax_c \frac{a_c}{\bar{a}_c}$. 
\end{definition}

The surprisingly popular method shows a certain amount of robustness to people's strategies compared to plurality vote. For example, if a large set of people always report one option with 100\% probability, the plurality vote method will be distorted but the surprising popular method will decrease this distortion. In fact, in the binary question case, we will show that the surprisingly popular method is equivalent to DMI-clustering thus is robust to people's strategies. However, in general, without any prior assumption, the surprisingly popular method does not satisfy the legal nor affine-invariant property, thus is not robust to people's strategies even if the number of agents goes to infinite. We will show a visual comparison between DMI-clustering and the surprisingly popular method for three-options case in Figure~\ref{fig:svsd}.

\begin{corollary}[Binary question: \dmi-clustering = Surprisingly popular]

	In the binary case, \dmi-clustering will map $\mathbf{a}=(a_1,a_2)$ to cluster $i$ if $a_i>\bar{a}_i$ for $i=1,2$. 
\end{corollary}
\begin{proof}
	In this case, $\mathbf{A}$ is a $n\times 2$ row-stochastic matrix. Since \dmi-clustering is affine-invariant, we subtract the row average from each row $\dmic(\mathbf{A})=\dmic(\mathbf{A}-\bar{\mathbf{a}})$. Note that $[\mathbf{A}-\bar{\mathbf{a}}\quad \mathbf{1}]$ is rank 2, we pick $\tilde{\mathbf{A}}$ as the first and third columns $[\mathbf{A}_{:,1}-\bar{a}_1\quad \mathbf{1}]$ and use $\mathbf{v}$ as a shorthand for $\mathbf{A}_{:,1}-\bar{a}_1$. Note that the sum of $\mathbf{v}$ is 0. Then by the same analysis in Corollary~\ref{coro:1d}, the optimal 0-1 vector $\mathbf{p}$ will pick all positive entries of $\mathbf{v}$. Thus, \dmi-clustering will map $\mathbf{a}=(a_1,a_2)$ to cluster $i$ if $a_i-\bar{a}_i>0$ for $i=1,2$. 
\end{proof}

\begin{figure}
    \center
	\includegraphics[width = 12cm]{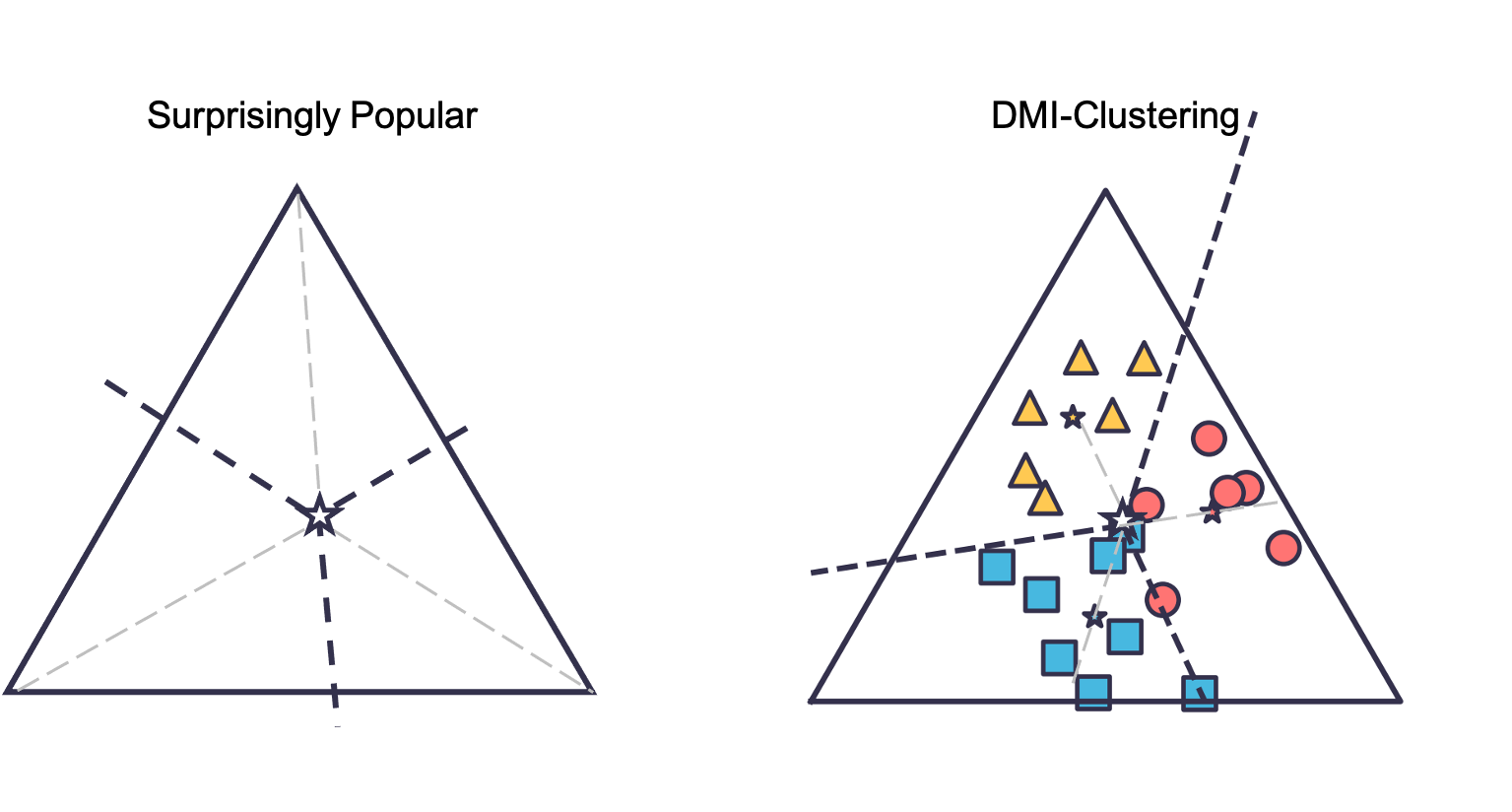}
	\caption{Surprisingly popular vs. DMI-clustering: both methods have the mean of all points as the center of the partitions. Left (surprisingly popular): the center determines the cluster, the partition lines (the black dashed lines) are generated by extending the line segment between the extreme point and the center. Right (DMI-clustering): the partition lines extend the line segments between each cluster's mean and the center. Thus, this partition can be seen as a ``rotated'' version of the surprisingly popular method's partition.}\label{fig:svsd}
\end{figure}

	Both surprisingly popular and DMI-clustering method set the mean of all points as the intersection of the partitions. We normalize each column of $\mathbf{A}$ by dividing it by its sum and will see that \dmi-clustering can be seen as a ``rotated'' version of surprisingly popular. 

\begin{corollary}[\dmi-clustering = ``rotated'' surprisingly popular]

With column-normalized $\mathbf{A}$, $\ncol(\mathbf{A})$, the surprisingly popular method is represented as $\mathbf{C}\leftarrow \idxmax(\ncol(\mathbf{A}))$. If $\ncol(\mathbf{A})$ is rank $C$, there exists $\mathbf{D}$ such that DMI-clustering's result is represented as $\mathbf{C}\leftarrow \idxmax(\ncol(\mathbf{A})\mathbf{D})$.
\end{corollary}

\begin{proof}
The first statement follows directly from the definition of surprisingly popular method. It's left to prove the second statement. Since DMI-clustering is invariant right affine transformation to $\mathbf{A}$, we can use column-normalized $\mathbf{A}$ as input. For $[\ncol(\mathbf{A}) \quad \mathbf{1}]$, the last column $\mathbf{1}$ is a linear combination of $\ncol(\mathbf{A})$'s columns since $\mathbf{A}$'s columns sum to $\mathbf{1}$. If $\ncol(\mathbf{A})$ is rank $C$, then we can pick $\tilde{\mathbf{A}}=\ncol(\mathbf{A})$. Thus, there exists $\mathbf{D}$ such that DMI-clustering's result is represented as $\mathbf{C}\leftarrow \idxmax(\ncol(\mathbf{A})\mathbf{D})$.

\end{proof}

\section{Single-task \& Method of Moments}

The biggest advantage of surprisingly popular is that it only uses the first moment information. Thus, it is originally proposed for the single-task setting. In this section, we will introduce the single-task setting and show that from a clustering view, we can utilize the second moment information to propose a spectral clustering based mechanism which can be better than surprising popular in some settings. 

\paragraph{Setting} Agents are assigned a single task (e.g. Have you ever text while driving before?). Agents finish the task without any communication. Each agent $i$ is asked to report her signal $c_i$ (e.g. No) and her prediction $\mathbf{p}_i$ for other people's signals (e.g. I think 10\% people say yes and 90\% people say no).

\begin{algorithm}[H]
\SetAlgoLined
\DontPrintSemicolon
\KwInput{$(c_i,\mathbf{p}_i)_i$ } 
\KwOutput{c}
$\mathbf{a}\in\mathbb{R}^{1\times C}$ where $\forall c, a_c$ is the ratio of people who answers $c$\;
$\forall c,c', \Pr[c'|c]=\frac{\sum_i\mathbbm{1}(c_i=c)\mathbf{p}_i(c')}{\sum_i\mathbbm{1}(c_i=c)}$\;
$\forall c, \Pr[c]=(\sum_{c'} \frac{\Pr[c'|c]}{\Pr[c|c']})^{-1}$, $\forall c,c', \Pr[c',c]=\Pr[c'|c] \Pr[c]$\tcc*{constructing first moment using the second moment}

$\mathbf{v}\in\mathbb{R}^{1\times C}$ where $\forall c, v_c:=\Pr[c]$\tcc*{first moment}
Return $\argmax_c \frac{a_c}{v_c}$

  \caption{Surprisingly popular (1st moment) \cite{prelec2017solution}}\label{alg:sp}
\end{algorithm}

\begin{algorithm}[H]
\SetAlgoLined
\DontPrintSemicolon
\KwInput{$(c_i,\mathbf{p}_i)_i$ } 
\KwOutput{$\plus$ or $\minus$}
Get answer collections $\mathbf{a}$, prior $\Pr[c],\forall c$, joint distribution $\Pr[c,c'],\forall c,c'$ like Algorithm \ref{alg:sp}\;
$\mathbf{v}\in\mathbb{R}^{1\times C}$ where $\forall c, v_c:=\Pr[c]$\tcc*{first moment}
$\mathbf{M}\in\mathbb{R}^{C\times C}$ where $\forall c,c', M_{c,c'}=\Pr[c,c']$\tcc*{second moment}
$\mathbf{Cov}=\mathbf{M}-\mathbf{v}^{\top}\mathbf{v}$\;
Get the first eigenvector $\mathbf{e}^*$ of $\mathbf{Cov}$\;
Return the sign of the inner product between $\mathbf{e}^*$ and $\mathbf{a}-\mathbf{v}$. 

  \caption{Spectral truth serum (2nd moment)}
\end{algorithm}

\begin{figure}
    \center
	\includegraphics[width = 11cm]{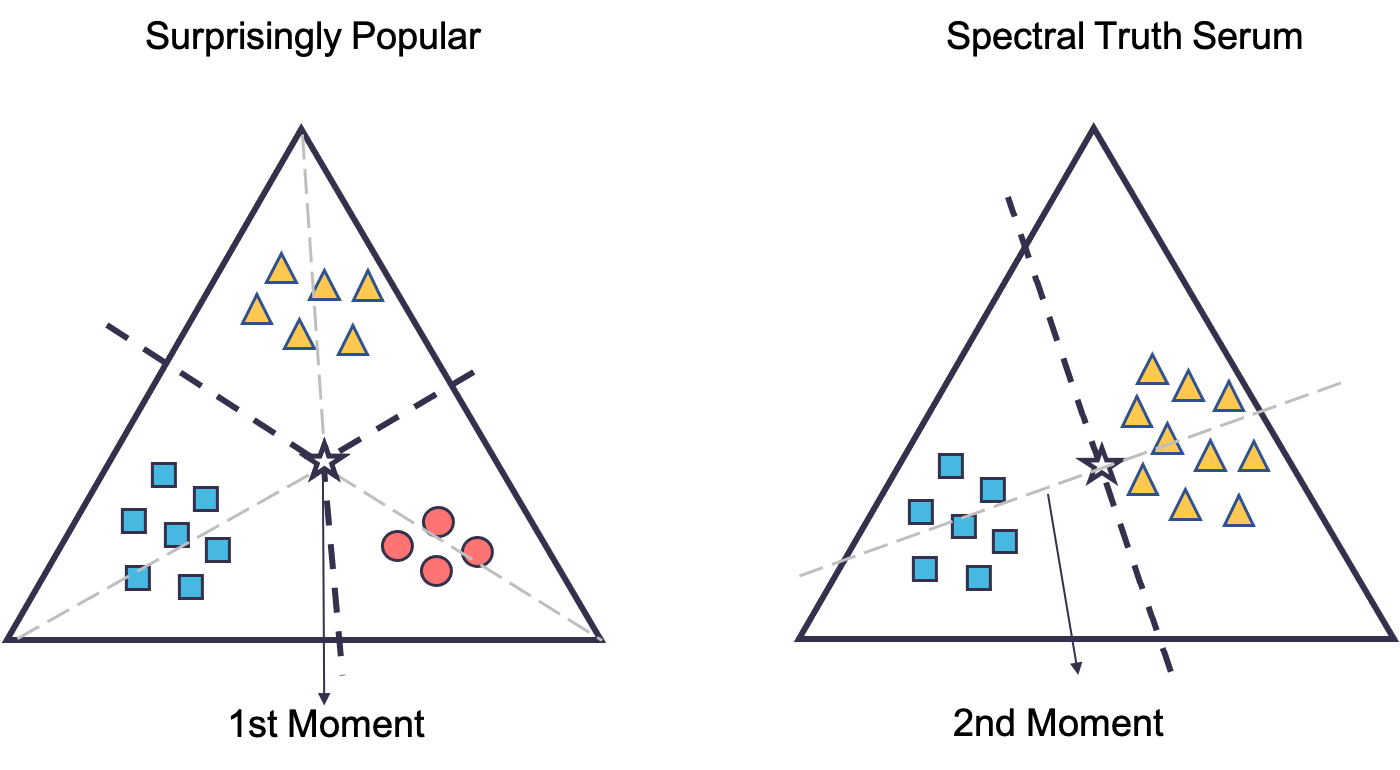}\label{fig:sts}
	\caption{\textbf{Surprisingly popular vs. Spectral Truth Serum}: the black dashed lines are partitions. In the right figure, the grey dashed line is $\mathbf{e}^*$. If the ground-truth clustering is more like the left case, the surprisingly popular method is better; if the ground-truth clustering is more like the right case (two worlds, spherical gaussians), the spectral truth serum is better.}
\end{figure}

\begin{lemma}\cite{blum2016foundations}
	The first eigenvector $\mathbf{e}^*$ of the covariance matrix of mixture of two spherical gaussians passes through the means. 
\end{lemma}

The above lemma shows that if the two Gaussian are sufficiently separated, we can project all points into $\mathbf{e}^*$ and set the overall mean $\mathbf{v}$ as the separator. This leads to spectral truth serum. Thus, in the scenarios that the ground-truth clustering is more like a well-separated mixture of two spherical Gaussian, spectral truth serum is better than surprisingly popular. Intuitively, spectral truth serum is applicable to the setting where there is only two underlying ground-truth states (e.g. fake news or not) and we elicit many signals from the crowds. The eigenvector of the covariance matrix tells us which signal is important and which signal is not. We then project people's feedback into the eigenvector, the important direction, and compare it to the average.

\citet{hsu2013learning} show that with a third-order moment, a mixture of $k$ gaussians can be detected. A future direction will be using a hypothetical question (e.g. will you update your prediction for other people's answers if you have seen some one answers A?) to elicit a third-order moment and obtain a better information aggregation results. 

\section{Conclusion and Discussion}

We provide a new clustering perspective for information elicitation. With this perspective, we propose a new information elicitation mechanism which is not only dominantly truthful but also aggregates information robustly when there are a large number of people in the multi-task setting. We also propose a spectral method based information aggregation mechanism in the single-task setting. Of independent interest, we propose a novel affine-invariant clustering method, DMI-clustering method. 

The geometric interpretation of DMI-clustering can be extended to any other $k\leq d$ easily.  For example, in a three-dimensional space, when we pick $k=2$, the clustering goal is to find two clusters such that the distance between their mean, multiplying the product of cluster sizes is maximized. We know that picking other $k$ will lose the affine-invariant property but we can still explore the properties of other $k$ in the future. 

Another interesting future direction is to apply other clustering techniques like non-linear clustering techniques (e.g. kernel method) to information elicitation. This work focuses on discrete setting and we can also consider continuous setting or study new information elicitation scenarios to provide new settings for clustering. Empirical validation is also an important future direction.

\bibliographystyle{ACM-Reference-Format}
\bibliography{ref}

\appendix 

\section{Synthetic Data Used for Demonstration}

\paragraph{Affine-invariance in Figure~\ref{fig:prop}} 

\[\mathbf{A}=\begin{bmatrix}
  0.96536243 & 0.83582504\\
  0.02223537 & 0.97962069\\
  0.18576474 & 0.09306992\\
  0.60073919 & 0.06909198\\
  0.21115965 & 0.04303247\\
  0.24518684 & 0.46305449\\
  0.0045001 & 0.73335878\\
\end{bmatrix}
\quad 
\mathbf{T}=\begin{bmatrix}
  0.7384872 & 0.39051911\\
  0.75115812 & 0.90684574\\
\end{bmatrix} \quad \mathbf{b} = b = [.05,.02]\]

\paragraph{An empirical run of k-cofactors in Figure~\ref{fig:empirical}} 

\[\mathbf{A}=\begin{bmatrix}
  0.32786192 & 0.75198752\\
  0.24789876 & 0.40110656\\
  0.64860833 & 0.05877561\\
  0.13688932 & 0.89837639\\
  0.51958647 & 0.69553312\\
  0.57106034 & 0.39534018\\
  0.25602628 & 0.5299542\\
  0.66521762 & 0.25370157\\
  0.49647253 & 0.51707767\\
  0.00649304 & 0.78141629\\
  0.93642378 & 0.00550147\\
  0.08014735 & 0.94850578\\
  0.54978178 & 0.66296321\\
  0.57201864 & 0.7952191\\
  0.11499522 & 0.59186407\\
  0.11681198 & 0.90701099\\
  0.01623413 & 0.59313684\\
  0.31841007 & 0.43924738\\
  0.60377183 & 0.96489354\\
  0.48738826 & 0.17097179\\
  0.49460288 & 0.81759015\\
  0.27831212 & 0.76869342\\
  0.22953897 & 0.43809347\\
  0.03654411 & 0.6203217\\
  0.90628651 & 0.45924219\\
  0.92525839 & 0.69213278\\
  0.17723574 & 0.12487727\\
  0.21277346 & 0.34931579\\
  0.84758762 & 0.05452689\\
  0.65204294 & 0.82808225\\
\end{bmatrix}
\]

\paragraph{DMI-clustering in Figure~\ref{fig:svsd}}

\[\mathbf{A}=\begin{bmatrix}
  0.20727033 & 0.56209307 & 0.2306366\\
  0.55235357 & 0.26311054 & 0.18453589\\
  0.06826729 & 0.51504916 & 0.41668355\\
  0.40863481 & 0.45383908 & 0.13752611\\
  0.30463115 & 0.19875226 & 0.49661659\\
  0.40387463 & 0.1261351 & 0.46999028\\
  0.30097293 & 0.32668147 & 0.3723456\\
  0.0530016 & 0.41863456 & 0.52836383\\
  0.53632014 & 0.08514494 & 0.37853491\\
  0.34180366 & 0.57497414 & 0.08322221\\
  0.20603456 & 0.67360041 & 0.12036503\\
  0.35331477 & 0.28461216 & 0.36207308\\
  0.23868633 & 0.38453375 & 0.37677992\\
  0.098419 & 0.67327932 & 0.22830168\\
  0.33219443 & 0.0159078 & 0.65189777\\
  0.40376774 & 0.39908311 & 0.19714915\\
  0.08966334 & 0.40876422 & 0.50157244\\
  0.51224362 & 0.01623191 & 0.47152447\\
  0.04907822 & 0.30059226 & 0.65032952\\
  0.50445751 & 0.20957023 & 0.28597227\\
\end{bmatrix}
\]

\end{document}